\documentclass[runningheads,a4paper]{llncs}[2015/06/24]

\usepackage{graphicx}

\usepackage[ngerman,english]{babel}
\addto\extrasenglish{\languageshorthands{ngerman}\useshorthands{"}}

\usepackage[%
rm={oldstyle=false,proportional=true},%
sf={oldstyle=false,proportional=true},%
tt={oldstyle=false,proportional=true,variable=true},%
qt=false%
]{cfr-lm}
\usepackage[math]{blindtext}

\usepackage{cite}

\usepackage{paralist}

\usepackage{csquotes}

\usepackage{microtype}

\usepackage{url}
\makeatletter
\g@addto@macro{\UrlBreaks}{\UrlOrds}
\makeatother

\usepackage{xcolor}

\usepackage{pdfcomment}
\hypersetup{hidelinks,
   colorlinks=true,
   allcolors=black,
   pdfstartview=Fit,
   breaklinks=true}
\usepackage[all]{hypcap}

\usepackage{xspace}

\DeclareFontFamily{U}{MnSymbolC}{}
\DeclareSymbolFont{MnSyC}{U}{MnSymbolC}{m}{n}
\DeclareFontShape{U}{MnSymbolC}{m}{n}{
    <-6>  MnSymbolC5
   <6-7>  MnSymbolC6
   <7-8>  MnSymbolC7
   <8-9>  MnSymbolC8
   <9-10> MnSymbolC9
  <10-12> MnSymbolC10
  <12->   MnSymbolC12%
}{}
\DeclareMathSymbol{\powerset}{\mathord}{MnSyC}{180}

\usepackage{amsmath}
\usepackage{amsfonts}
\usepackage{tikz}
\usetikzlibrary{calc, arrows, positioning, backgrounds, shapes, fit, hobby, decorations.pathmorphing, decorations.pathreplacing}
\usepackage{hyperref}
\usepackage{tabularx,environ}
\usepackage{paralist}
\usepackage{comment}
\usepackage{subcaption}

\makeatletter
\newcommand\footnoteref[1]{\protected@xdef\@thefnmark{\ref{#1}}\@footnotemark}
\makeatother

\makeatletter
\newcommand{\problemtitle}[1]{\gdef\@problemtitle{#1}}
\newcommand{\probleminput}[1]{\gdef\@probleminput{#1}}
\newcommand{\problemquestion}[1]{\gdef\@problemquestion{#1}}
\newcommand{\problemoutput}[1]{\gdef\@problemoutput{#1}}
\NewEnviron{prob}{
	\problemtitle{}\probleminput{}\problemquestion{}
	\BODY
	\par\addvspace{.5\baselineskip}
	\noindent
	\begin{tabularx}{\textwidth}{@{\hspace{\parindent}} l X c}
		\multicolumn{2}{@{\hspace{\parindent}}l}{\@problemtitle} \\
		\textbf{Input:} & \@probleminput \\
		\textbf{Question:} & \@problemquestion
	\end{tabularx}
	\par\addvspace{.5\baselineskip}
}
\NewEnviron{optproblem}{
	\problemtitle{}\probleminput{}\problemoutput{}
	\BODY
	\par\addvspace{.5\baselineskip}
	\noindent
	\begin{tabularx}{\textwidth}{@{\hspace{\parindent}} l X c}
		\multicolumn{2}{@{\hspace{\parindent}}l}{\@problemtitle} \\
		\textbf{Input:} & \@probleminput \\
		\textbf{Output:} & \@problemoutput
	\end{tabularx}
	\par\addvspace{.5\baselineskip}
}
\makeatother

\renewcommand{\vec}[1]{\boldsymbol{#1}}

\newtheorem{innernumberedlemma}{Lemma}
\newenvironment{numberedlemma}[1]
{\renewcommand\theinnernumberedlemma{#1}\innernumberedlemma}
{\endinnernumberedlemma}

\begin{document}

\title{Lexico-minimum Replica Placement in Multitrees}

\author{K. Alex Mills \and R. Chandrasekaran \and Neeraj Mittal}
\institute{Department of Computer Science,	\\University of Texas at Dallas,\\ Richardson, TX, 
	USA}

%

\maketitle

\begin{abstract}
In this work, we consider the problem of placing replicas in a data center or storage area network, represented as a digraph, so as to lexico-minimize a previously proposed reliability measure which minimizes the impact of all failure events in the model in decreasing order of severity. Prior work focuses on the special case in which the digraph is an arborescence. In this work, we consider the broader class of multitrees: digraphs in which the subgraph induced by vertices reachable from a fixed node forms a tree. We parameterize multitrees by their number of ``roots'' (nodes with in-degree zero), and rule out membership in the class of fixed-parameter tractable problems (FPT) by showing that finding optimal replica placements in multitrees with 3 roots is NP-hard. On the positive side, we show that the problem of finding optimal replica placements in the class of \emph{untangled} multitrees is FPT, as parameterized by the replication factor $\rho$ and the number of roots $k$. Our approach combines dynamic programming (DP) with a novel tree decomposition to find an optimal placement of $\rho$ replicas on the leaves of a multitree with $n$ nodes and $k$ roots in $O(n^2\rho^{2k+3})$ time.
\end{abstract}


\section{Introduction}

As data centers become larger, ensuring reliable access to the data they store becomes a greater concern. Each piece of hardware introduces a new point of failure -- the more hardware, the more likely it is that failure will occur. Moreover, to keep large-scale data centers cost-effective, they are typically built using commodity hardware, further increasing the likelihood of a failure event. Ensuring the availability and responsiveness of data center operations in such environments has been a subject of recent interest.

Many availability problems are solved through the use of replication: placing identical copies of data or tasks across multiple machines to ensure the survival of one replica in case of failure. While this approach has been known for decades, researchers have recently begun to cast the specific problem of replica placement as an optimization problem in which the dependencies among failure events are modeled \cite{Korupolu2016, Mills2015}. To date, these approaches have relied on the simplifying assumption that the failure event model is hierarchically arranged. While such models coincide with some real-world systems \cite{Parallels, VMWare}, providing optimal replica placements for more general models remains an interesting problem.

Of special interest is the measurement used to score the reliability of a placement. One standard approach involves assigning to each failure event its likelihood of occurrence. But this approach is subject to the following critiques. First, measurements or estimations of failure probability may themselves be unreliable -- thereby providing an unreliable basis for optimization. Second, even a perfect measurement of failure based on historical behavior cannot account for a failure pattern which has never occurred before, and therefore could not have been measured. In other words: ``past performance is not an indicator of future results''.

In light of these concerns, we have proposed in \cite{Mills2015} a multi-criteria reliability measure which minimizes the impact of failure events in the aggregate. Specifically, we introduce a reliability metric which places failure events into buckets based on their \emph{impact} -- the number of replicas which they cause to become unavailable. We then minimize the number of events in each bucket in decreasing order of impact. As a result, the placements we obtain achieve the minimum number of events which cause all replicas to fail (i.e. the number of events with maximum impact). Subject to this quantity being minimized, we then minimize the number of events which cause all but one replica to fail, followed by minimizing the number of events that cause all but two replicas to fail, and so on.

This goal is achieved by minimizing a vector quantity called the \emph{failure aggregate} in the lexicographic order. Our past work investigates minimizing failure aggregates of replicas placed on the leaves of a tree. For this problem an $O(n + \rho \log \rho)$ algorithm can be achieved, where $n$ is the number of nodes in the tree, and $\rho$ is the number of replicas to be placed \cite{Mills2015}. We have also investigated simultaneously minimizing \emph{multiple} placements on the leaves of a tree \cite{Mills2017}. Our current solution to this problem runs in polynomial time when the \emph{skew} is constant. The skew is defined as the maximum absolute difference in number of replicas placed among all pairs of placements. For a skew of $\delta$, we present an algorithm to place $m$ groups of replicas on the leaves of a tree with $n$ nodes in $\tilde{O}(n\rho^3\delta^3m^\Delta / \Delta!)$ time where $\rho$ is the maximum number of replicas placed among all $m$ groups, and $\Delta = O(\delta^2)$ \cite{Mills2017}.

While some commercially available storage area networks use failure domains modeled by trees \cite{Parallels,VMWare}, extensions to more general failure domain models are an important research goal. In this work, we initiate the parameterized study of the problem of lexico-minimum replica placement in multitrees, as parameterized by the number of its roots. A \emph{multitree} is defined as a directed acyclic graph (DAG) in which, for any fixed vertex $v$, the set of vertices reachable from $v$ forms a tree as an induced subgraph. The \emph{roots} and \emph{leaves} of a multitree are defined as nodes with in-degree zero and out-degree zero respectively. We emphasize the parameter by referring to a multitree with $k$ roots as a $k$-multitree.  Our goal is to place $\rho$ replicas on the leaves of a $k$-multitree so that the failure aggregate is minimized in the lexicographic order. 

We show that lexico-minimum replica placement is NP-hard even in $3$-multitrees, ruling out fixed-parameter tractability for this parameterization. The proof we present relies on the Four Color Theorem \cite{Appel1977} to exploit a disparity in hardness of two well-known problems restricted to cubic planar bridgeless graphs. In such graphs, finding a 3-edge-coloring can be done in polynomial time, while solving \textsc{independent set} remains NP-hard. To circumvent this hardness result, we define \emph{untangled} multitrees, a class of multitrees for which we exhibit membership in FPT. We develop a FPT algorithm based on the tree decomposition approach. Since multitrees are a special case of directed acyclic graphs, standard decomposition approaches such as treewidth \cite{Bodlaender2008}, pathwidth \cite{Andreica2008}, and DAG-width \cite{Berwanger2012} do not apply. Instead, we provide a novel decomposition technique tailored to our problem.

Our algorithm works in two successive phases, a \emph{decomposition phase} and an \emph{optimization phase}. The decomposition phase produces a specialized \emph{decomposition tree}, a full\footnote{Recall that in a full binary tree every node has 0 or 2 children.} binary tree in which each node is associated with an induced subgraph of the input multitree. The optimization phase then runs a bottom-up dynamic programming algorithm over the nodes of the decomposition tree. While the overall process is similar to FPT algorithms for graphs with restricted treewidth, our decomposition technique and application are both novel. Our algorithm for untangled $k$-multitrees runs in $O(n^2\rho^{2k+3})$ time, thus demonstrating that lexico-minimum replica placement on untangled $k$-multitrees is in FPT, as parameterized by $\rho$ \emph{and} $k$. 

\section{Modeling Reliable Replica Placement in Multitrees}

In this section we formalize the model presented in the introduction. We model the failure domains of a data center as a \emph{multitree}, a directed acyclic graph (DAG) whose formal definition we defer to the next paragraph. Non-leaf vertices represent failure events which are typically associated with the failure of a physical hardware component, but may instead be associated with abstract events such as network maintenance or software failures. Leaf vertices represent servers on which replicas of data may be placed. A directed edge between two failure events $u$ and $v$ indicates that the failure of event $u$ may trigger failure event $v$.

A \emph{multitree} is a directed acyclic graph (DAG) in which the set of vertices reachable from any vertex forms an arboresence (see Fig. 1(a)). In the context of graph $G$, let $u \rightsquigarrow v$ denote the assertion ``there is a path from $u$ to $v$ in $G$'', and $u \rightarrow v$ denote the assertion ``there is an edge from $u$ to $v$ in G''. Then a multitree is equivalently defined as a \emph{diamond-free} DAG \cite{Furnas1994}. See Fig. 1(b) for a depiction of the forbidden subgraphs used to define diamond-free DAGs below.
\begin{definition}
	A \emph{multitree} $M = (V,E)$ is a DAG in which there are no diamonds (i.e. a DAG which is \emph{diamond-free}). A \emph{diamond} is either \begin{inparaenum}[(1)] \item a set of three vertices $a,b,c \in V$ in which $a \rightarrow b \rightsquigarrow c$, and, even when the edge $(a,b)$ is removed, $a \rightsquigarrow c$, or \item a set of four vertices $a,b,c,d \in V$ in which $a \rightsquigarrow b \rightsquigarrow d$ and $a \rightsquigarrow c \rightsquigarrow d$, while there is no path from $b$ to $c$ or vice versa.\end{inparaenum}
\end{definition}

\begin{figure}
	\captionsetup[subfigure]{justification=centering}
	\centering
	\begin{subfigure}{0.38\linewidth}
		\centering
		\includegraphics[scale=0.4]{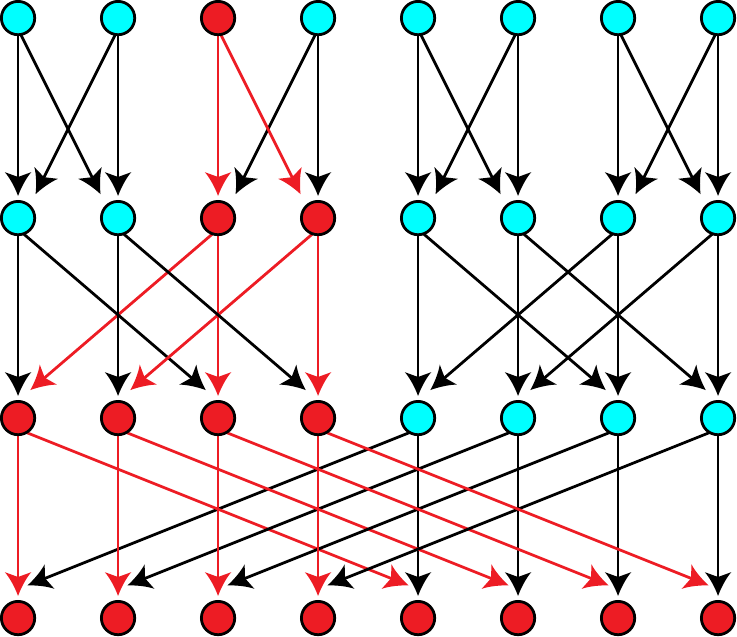}
		\subcaption{}\label{f:multitree-subfig}
	\end{subfigure}
	\hfill
	\begin{subfigure}{0.5\linewidth}
		\centering
		\begin{tikzpicture}[->, >=stealth',every node/.style={minimum height=3mm, draw, circle, node distance=4.5mm, fill=white},scale=0.8, transform shape]

\tikzstyle{snakeline} = [decorate, decoration={pre length=0.1cm,
	post length=0.1cm, snake, amplitude=.4mm,
	segment length=2mm}, ->, >=stealth']

\node (a) {$a$};
\node [below =1cm of a] (b) {$c$};
\node [right = of a, yshift=-0.75cm] (c) {$b$};

\draw (a) edge[snakeline] (b)
	  (a) edge (c)
	  (c) edge[snakeline] (b);

\node [right = 4cm of a, yshift=0.5cm](a2) {$a$};
\node [left = of a2, yshift=-1cm](b2) {$b$};
\node [right = of a2, yshift=-1cm](c2) {$c$};
\node [below = 1.5cm of a2](d2) {$d$};

\draw (a2) edge[snakeline] (b2)
	(a2) edge[snakeline] (c2)
	(b2) edge[snakeline] (d2)
	(c2) edge[snakeline] (d2);

\end{tikzpicture}
		\subcaption{}\label{f:diamonds-forbidden-subgraphs}
	\end{subfigure}
	\caption{(\subref{f:multitree-subfig}) A multitree in which red highlights depict an induced subgraph forming an arboresence, (\subref{f:diamonds-forbidden-subgraphs}) Forbidden subgraphs in which squiggles depict an arbitrary path.}\label{f:multitree-fig}
\end{figure}

A $k$-multitree is a multitree with $k$ roots. In context of a multitree $M = (V,E)$ we denote the set of leaves of $M$ by $L \subseteq V$. In context of our problem we seek a subset of leaves on which to place replicas of data. To this end, we define a \emph{placement} of $\rho$ replicas as a subset\footnote{Using a subset as opposed to a multiset rules out the possibility of placing multiple replicas on the same server, which would defeat the purpose of replication.} of leaves $P \subseteq L$ with size $|P| = \rho$.

Given a placement $P$, we associate to each failure event its \emph{failure number}: the number of replicas from $P$ which can be made unavailable should the event occur. The failure number of $u$ is equal to the number of nodes in $P$ which are reachable from $u$, which we denote as $f(u,P) := |\{x \in P : u \rightsquigarrow x \}|$.

To aggregate the failure numbers across all failure events into a single vector-valued quantity, we denote the \emph{failure aggregate} by $\vec{f}(P) = \langle p_0, p_1, ..., p_{\rho}\rangle$, where $p_i = |\{u \in V : f(u,P) = \rho - i\}|$. Intuitively, the $i^{th}$ entry of $\vec{f}(P)$ contains the number of events whose failure leaves $i$ replicas surviving.

Our optimization goal is to minimize the failure aggregate in the lexicographic order, which was motivated in the introduction. The (strict) lexicographic order $<_L$ between vectors $\vec{x} = \langle x_0,...,x_n \rangle$ and $\vec{y} = \langle y_0,...,y_n \rangle$ is defined via the formula \[x <_L y \iff \exists j \in [0, n] : (x_j < y_j \wedge \forall i < j [x_i = y_i]),\]
while the weak lexicographic order $\leq_L$ is defined by extending $<_L$ in the usual way. We use the short-hand ``lexico-minimum'' and ``lexico-minimizes'' to mean ``minimum'' and ``minimizes'' in the lexicographic order respectively.

With these definitions in hand, we provide the formal definition of the parameterized optimization problem we consider in the remainder of this paper.
\begin{optproblem}
	\problemtitle{\textsc{Lexico-minimum Single-block Placement in $k$-Multitrees} ($k$-LSP)}
	\probleminput{A $k$-multitree, $M=(V,E)$; the set of leaves $L \subseteq V$; a positive integer $\rho < |L|$}
	\problemoutput{A placement $P \subseteq L$ with $|P| = \rho$ such that $\vec{f}(P)$ is lexico-minimum among all placements $P \subseteq L$ with $|P| = \rho$.}
\end{optproblem}

\section{NP-hardness of 3-LSP} 

In this section, we concern ourselves with how the hardness of $k$-LSP depends on the parameter $k$.  Prior work has shown that $1$-LSP can be solved in polynomial time \cite{Mills2015}, since a $1$-multitree is just an arboresence. In this section we show that $3$-LSP is NP-hard, thereby ruling out a fixed-parameter tractable algorithm parameterized by the number of roots.

Specifically, we show hardness of the following decision problem.
\begin{prob}
	\problemtitle{\textsc{Lexicographic Replica Placement in $3$-multitrees} ($3$-LSP)}
	\probleminput{A $3$-multitree, $M=(V,E)$ with leaves $L \subseteq V$; a positive integer $\rho$; and a vector $\vec{w} \in \mathbb{N}^{\rho+1}$}
	\problemquestion{Is there a placement $P \subseteq L$ with $|P| = \rho$ such that $\vec{f}(P) \leq_L \vec{w}$?}
\end{prob}

We will prove that this problem is NP-hard by reduction from \textsc{Independent Set} restricted to cubic planar bridgeless graphs. Cubic planar bridgeless graphs are guaranteed to have a 3-edge-coloring \cite{Goemans2012}. Moreover, 3-coloring the edges of such graphs is equivalent to 4-coloring their faces \cite{Tait1880b}. The faces of such graphs correspond to the vertices of a planar graph, and, as a consequence of the Four Color Theorem, finding a 4-vertex-coloring of a planar graph may be done in $O(n^2)$ time \cite{Cole2008}. On the other hand, finding an independent set in such graphs is NP-hard, as was shown in \cite{Mohar2001}. We exploit the disparity in the hardness of these two problems to show that $3$-LSP is NP-hard, by reduction from the following problem.

\begin{prob}
	\problemtitle{\textsc{Restricted Independent Set} (RIS)}
	\probleminput{An undirected cubic planar bridgeless graph $G = (V,E)$; a positive integer $k$.}
	\problemquestion{Does $G$ admit an independent set of size exactly $k$?}
\end{prob}

\begin{theorem}\label{t:3-multitree-np-hard}
	RIS reduces to $3$-LSP in polynomial time. Thus, $3$-LSP is NP-hard.
\end{theorem}

\begin{proof}
	Given a cubic planar bridgeless graph $G=(V,E)$, we can form a $3$-multitree, $H$, as follows. Let $H = (V',E')$. Add a vertex to $H$ for every edge in $E$ and for every vertex in $V$. Let the vertices of $H$ that represent vertices of $G$ be denoted by $H(V)$ and let the vertices of $H$ that represent edges of $G$ be denoted by $H(E)$. Next, for every edge $e = (u,v)$ of $G$, add directed edges $(e,u)$ and $(e,v)$ to $H$.  Next, we partition $H(E)$ into three sets, $S_1, S_2, S_3$, such that no node in $H(V)$ has two neighbors in the same set. This partition corresponds to finding a $3$-edge-coloring of $G$, which may be done in $O(n^2)$ time \cite{Cole2008}. We then add three special nodes $\alpha, \beta$ and $\gamma$ to $H$, and add edges $(\alpha,s_1), (\beta,s_2), (\gamma,s_3)$ for all $s_1 \in S_1, s_2 \in S_2 $ and $s_3 \in S_3$.

	\begin{figure}
		\centering
		\hfill
		\begin{subfigure}{0.38\linewidth}
			\centering
			\begin{tikzpicture}[every node/.style={minimum height=3mm, draw, circle, node distance=4.5mm, fill=white},scale=0.4]
	
	\node (a) at (0,0) {$a$};
	\node (b) at (4.5,3) {$b$};
	\node (c) at (4,-2) {$c$};
	\node (d) at (2, -4) {$d$};
	\node (f) at (8,-2) {$f$};
	\node (g) at (6, 0) {$g$};
	\node (e) at (11,-3) {$e$};
	\node (h) at (10,2)  {$h$};
	
		\path 
		(a) edge[ultra thick, color=cyan] (b) 
		(a) edge[ultra thick, color=green] (c)
		(a) edge[ultra thick, color=red] (d)
		(c) edge[ultra thick, color=red] (f)
		(c) edge[ultra thick, color=cyan] (d)
		(f) edge[ultra thick, color=cyan] (e)
		(g) edge[ultra thick, color=green] (f)
		(d) edge[ultra thick, color=green] (e)
		(b) edge[ultra thick, color=red] (g)
		(h) edge[ultra thick, color=green] (b)
		(e) edge[ultra thick, color=red] (h)
		(g) edge[ultra thick, color=cyan] (h);
		
		\node[draw=none,fill=none] at ($(a.center)!0.5!(b.center)$) {1};
		\node[draw=none,fill=none] at ($(b.center)!0.5!(g.center)$) {7};
		\node[draw=none,fill=none] at ($(g.center)!0.5!(h.center)$) {9};
		\node[draw=none,fill=none] at ($(b.center)!0.5!(h.center)$) {12};
		\node[draw=none,fill=none] at ($(g.center)!0.5!(f.center)$) {8};
		\node[draw=none,fill=none] at ($(h.center)!0.5!(e.center)$) {11};
		\node[draw=none,fill=none] at ($(e.center)!0.5!(f.center)$) {10};
		\node[draw=none,fill=none] at ($(d.center)!0.5!(e.center)$) {5};
		\node[draw=none,fill=none] at ($(c.center)!0.5!(f.center)$) {6};
		\node[draw=none,fill=none] at ($(d.center)!0.5!(c.center)$) {4};
		\node[draw=none,fill=none] at ($(a.center)!0.5!(d.center)$) {3};
		\node[draw=none,fill=none] at ($(a.center)!0.5!(c.center)$) {2};

\end{tikzpicture}
		\end{subfigure}
		\hfill
		\begin{subfigure}{0.5\linewidth}
			\centering
			\begin{tikzpicture}[->, >=stealth',every node/.style={minimum height=2mm, draw, circle, inner sep = 0mm,node distance=3mm, fill=black}, every edge/.style={draw=gray!80!white},scale=0.6]

	\node [label={1}] (1) {};
	\node [label={2},right=of 1] (2) {};
	\node [label={3},right=of 2] (3) {};
	\node [label={4},right=of 3] (4) {};
	\node [label={5},right=of 4] (5) {};
	\node [label={6},right=of 5] (6) {};
	\node [label={7},right=of 6] (7) {};
	\node [label={8},right=of 7] (8) {};	
	\node [label={9},right=of 8] (9) {};	
	\node [label={10},right=of 9] (10) {};		
	\node [label={11},right=of 10] (11) {};
	\node [label={12},right=of 11] (12) {};
	
	\node [below=1cm of 6, label=270:{$d$}] (d) {};
	\node [left=of d, label=270:{$c$}] (c) {};
	\node [left=of c, label=270:{$b$}] (b) {};
	\node [left=of b, label=270:{$a$}] (a) {};
	\node [right=of d, label=270:{$e$}] (e) {};
	\node [right=of e, label=270:{$f$}] (f) {};
	\node [right=of f, label=270:{$g$}] (g) {};
	\node [right=of g, label=270:{$h$}] (h) {};
	
	\node [above=1cm of 6, xshift=2mm, label={$\beta$}, fill=red] (beta) {}; 
	\node [left=1cm of beta, label={$\alpha$}, fill=cyan] (alpha) {}; 	
	\node [right=1cm of beta, label={$\gamma$},fill=green] (gamma) {}; 	
	
	\begin{scope}[on background layer]
		
		\draw (1) edge (a)
		(2) edge (a)
		(3) edge (a)
		(1) edge (b)
		(7) edge (b)
		(12) edge (b)
		(2) edge (c)
		(4) edge (c)
		(6) edge (c)
		(3) edge (d)
		(4) edge (d)
		(5) edge (d)
		(5) edge (e)
		(10) edge (e)
		(11) edge (e)
		(6) edge (f)
		(8) edge (f)
		(10) edge (f)
		(7) edge (g)
		(8) edge (g)
		(9) edge (g)
		(9) edge (h)
		(11) edge (h)
		(12) edge (h)
		;
		
		\draw (alpha) edge (1)
		(alpha) edge (4)
		(alpha) edge (9)
		(alpha) edge (10)
		
		(beta) edge (3)
		(beta) edge (7)
		(beta) edge (6)
		(beta) edge (11)
		
		(gamma) edge (2)
		(gamma) edge (5)
		(gamma) edge (8)
		(gamma) edge (12)
		;
	\end{scope}

		\draw 
			(gamma) edge[green!70!black, thick] (2)
			(gamma) edge[green!70!black, thick] (5)
			(gamma) edge[green!70!black, thick] (8)
			(gamma) edge[green!70!black, thick] (12)
			(2) edge[green!70!black, thick] (a)
			(12) edge[green!70!black, thick] (b)
			(2) edge[green!70!black, thick] (c)
			(5) edge[green!70!black, thick] (d)
			(5) edge[green!70!black, thick] (e)
			(8) edge[green!70!black, thick] (f)
			(8) edge[green!70!black, thick] (g)
			(12) edge[green!70!black, thick] (h)
		;

\end{tikzpicture}
		\end{subfigure}
		\hfill
		\caption[Illustration of the NP-hardness reduction for LSP in 3-multitreess.]{The $3$-edge-colored cubic planar bridgeless graph on the left maps to the 3-multitree on the right via our reduction. The roots of the 3-multitree correspond to the three color classes used in the 3-edge coloring on the left. On the right, the subtree induced by descendants of $\gamma$ is highlighted.}\label{f:3-multitree-reduction}
	\end{figure}
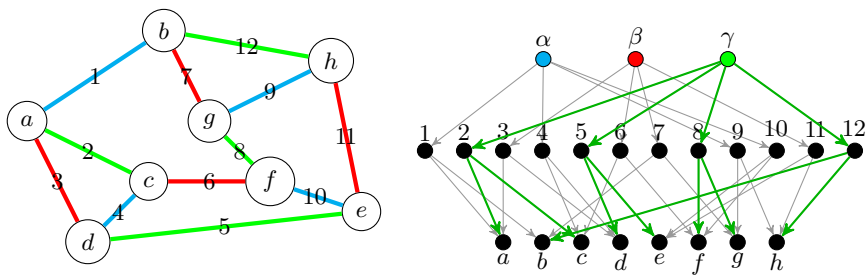
	
	We claim that $H$ is a 3-multitree. $H$ clearly has only three nodes with in-degree zero, so it suffices to show that no diamond is formed. Three-node diamonds are clearly impossible by construction. Instead suppose that there are vertices $a,b,c,d$ of $H$ which form a four-node diamond (i.e., $(a,b)(a,c)(b,c)(c,d) \in E'$). By construction, $d$ must be a node in $H(V)$, thus $b$ and $c$ must be nodes in $H(E)$, and $a = \chi$ for some $\chi \in \{\alpha,\beta,\gamma\}$, all of which follows from our construction. But then $d$ is a vertex in $H(V)$ which has two of its neighbors connected to the same root node $\chi$, a contradiction. Hence no diamond is created and $H$ is a $3$-multitree.
	
	Since each node in $G$ must be adjacent to an edge from each color class, every node in $H(V)$ must have $\alpha,\beta$ and $\gamma$ as ancestors. Thus, each of $\alpha,\beta$ and $\gamma$ have failure number $\rho$ in any placement of size $\rho$ on the leaves of $H$. Finally, we complete the reduction by showing that $H$ has a placement $P\subseteq H(V)$ with $|P| = k$ for which $\vec{f}(P) \leq_L \langle 3,0,...,0,\infty,\infty \rangle$ if and only if $G$ has an independent set of size $k$. This portion of the proof is straight-forward, and has been moved to Appendix \ref{a:omitted-proofs}. \qed
\end{proof}

Since it shows that, $k$-LSP is NP-hard even for a \emph{fixed} value of the parameter $k$, Theorem \ref{t:3-multitree-np-hard} rules out the existence of an FPT algorithm for $k$-multitrees as parameterized by the number of roots. Thus, $k$-LSP falls no lower in the $W$-hierarchy than $W[1]$. While a polynomial time algorithm for $1$-LSP was shown in \cite{Mills2015}, the complexity of $2$-LSP is open.

\section{Untangling Multitrees}

On the positive side, we show how a tree decomposition approach may be employed to yield an FPT algorithm for the subclass of \emph{untangled} $k$-multitrees. We use the term  \emph{connectors} to refer to vertices of a multitree which have in-degree strictly greater than 1. An untangled multitree is a multitree with additional requirements placed on the ancestry of connectors. Roughly speaking, we require that an untangled multitree may be split into two subgraphs such that a) the descendants of each non-root node fall into the same subgraph, and b) each connector is present in only \emph{one} of the two subgraphs. This property allows us to perform a decomposition of each multitree into two subgraphs. To make this idea precise, we employ the following modified notion of laminarity which we call a \emph{laminar pair} of set families.

\begin{definition}
	Two set families $\mathcal{F}, \mathcal{F}' \subseteq 2^X$ on the same ground set $X$ form a \emph{laminar pair} when, for all $U \in \mathcal{F}$, $V \in \mathcal{F}'$, either $U \subseteq V, U \supseteq V$, or $U \cap V = \emptyset$.
\end{definition}

To ensure the decomposability of a multitree $M=(V,E)$ into subgraphs $M_1$ and $M_2$, we require that for every child $c$ of each root, the set of connectors which are descendants of $c$ all lie in either $M_1$ or $M_2$. To formalize this idea, we define the \emph{connector shadow} as follows.
\begin{definition}
	Given a vertex $u \in V$, the \emph{connector shadow of $u$}, denoted $Sh(u)$, is the set of connectors of $M$ which are descendants of $u$.
\end{definition}
\begin{definition}
	Given a vertex $u \in V$, with children $c_1,...,c_m$, the \emph{child shadows} of $u$ is the set family defined as $\mathcal{C}(u) := \{Sh(c_1), ..., Sh(c_m)\}$.
\end{definition}
\begin{definition}
	Multitree $M = (V,E)$ is said to be \emph{untangled}  if, for every pair of vertices $u,v \in V$ where $u$ is not reachable from $v$ and vice versa, $\mathcal{C}(u)$ and $\mathcal{C}(v)$ are laminar pairs.
\end{definition}
Being untangled is easily seen to be a hereditary graph property\footnote{That is, if $M$ is an untangled multitree, then for every $U \subseteq V$, the vertex-induced subgraph $M[U] = (U, (U \times U) \cap E)$ is also an untangled multitree.}. 

While the class of untangled multitrees may appear to be highly specialized, it is in fact general enough to capture any directed acyclic graph. Any directed acyclic graph $G = (V,E)$ with leaves $L$ can be converted to a \emph{canonical placement model}, $H = (V,E')$, where
\[E' = \{(u,v) : u \in V \setminus L, v \in L, \text{ and } v \text{ is reachable from $u$ in $G$.}\}.\] See \autoref{f:canonicalplacements} for an example. By definition, the canonical placement model $H$ has the same reachability relation as the original graph $G$. This further implies that the failure numbers of placements on the leaves of $H$ have the same failure aggregate as their counterparts in $G$. Thus, a lexico-minimum placement in $H$ is also lexico-minimum in $G$. Furthermore, $H$ is easily seen to be a multitree, but also an \emph{untangled} multitree, since the set of child shadows for any vertex in $H$ is a family only containing singleton sets, and any pair of families of singleton sets trivially forms a laminar pair.

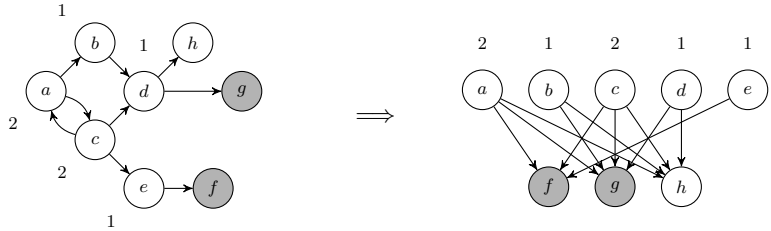
\begin{figure}[htb]
	\begin{subfigure}{0.45\linewidth}
		\centering 
		\begin{tikzpicture}[->,>=stealth',every node/.style={minimum height=7mm, draw, circle, node distance=5mm},scale=0.75,transform shape]
	\node [] (a) {$a$};
	\node [above right=of a] (b) {$b$};
	\node [below right= of a] (c) {$c$};
	\node [below right=of c] (e) {$e$};
	\node [above right=of c] (d) {$d$};
	\node [above right=of d] (h) {$h$};
	\node [fill=white!70!black,below right=of h] (g) {$g$};
	\node [fill=white!70!black,right=of e] (f) {$f$};
	
	\path 
	   (a) edge (b)
	   (b) edge (d)
	   (d) edge (h)
	   (d) edge (g)
	   (a) edge[bend left] (c)
	   (c) edge[bend left] (a)
	   (c) edge (e)
	   (c) edge (d)
	   (e) edge (f);
	   

		\node [draw=none, below left =1mm of a] {$2$};
		\node [draw=none, above left =1mm of b] {$1$};
		\node [draw=none, below left =1mm of c] {$2$};
		\node [draw=none, above =1mm of d] {$1$};
		\node [draw=none, below left =1mm of e] {$1$};
\end{tikzpicture}
	\end{subfigure}
	$\implies$
	\begin{subfigure}{0.45\linewidth}
		\centering
		\begin{tikzpicture}[->,>=stealth',every node/.style={minimum height=7mm, draw, circle, node distance=4.5mm},scale=0.75,transform shape]
	\node []           (a) {$a$};
	\node [right=of a] (b) {$b$};
	\node [right= of b] (c) {$c$};
	\node [right=of c] (d) {$d$};
	\node [right=of d] (e) {$e$};
	\node [fill=white!70!black,below=10mm of c] (g) {$g$};
	\node [right=of g]  (h) {$h$};
	\node [fill=white!70!black,left=of g] (f) {$f$};
	
	\path 
		(a) edge (f)
		(a) edge (g)
		(a) edge (h)
		(b) edge (g)
		(b) edge (h)
		(c) edge (f)
		(c) edge (g)
		(c) edge (h)
		(d) edge (g)
		(d) edge (h)
		(e) edge (f);
		
		\node [draw=none, above =1mm of a] {$2$};
		\node [draw=none, above =1mm of b] {$1$};
		\node [draw=none, above =1mm of c] {$2$};
		\node [draw=none, above =1mm of d] {$1$};
		\node [draw=none, above =1mm of e] {$1$};
\end{tikzpicture}
	\end{subfigure}
	\caption{A directed acyclic graph on the right, and the associated canonical placement model on the left. The highlighted nodes $\{f,g\}$ form a placement which induces the same failure numbers in both graphs.}\label{f:canonicalplacements}
\end{figure}
\vspace{-1cm}

\section{Decomposing $k$-multitrees}

As previously discussed, our algorithm runs in two sequential phases: a decomposition phase and an optimization phase. The decomposition phase of our algorithm takes as input a (weakly-connected) untangled $k$-multitree $M =(V,E)$ and produces as output a \emph{decomposition tree}. A decomposition tree is a full binary tree in which each node $u$ is associated with a subset of vertices of $M$ we call a \emph{subproblem}, denoted by $\Gamma_u \subseteq V$.

\begin{definition}\label{d:subproblem-tree}
	A \emph{decomposition tree} $\tau$ is a binary tree in which each node $u$ is associated with a subproblem $\Gamma_u \subseteq V$.
\end{definition}

\begin{definition}\label{d:trivial}
	A subproblem $\Gamma_u$ is said to be \emph{trivial} if $\Gamma_u$ contains no leaf nodes.
\end{definition}

To ensure that our decomposition preserves optimal substructure, we define the notion of an \emph{admissible} subproblem. In every decomposition tree produced by our procedure, internal nodes are associated with admissible subproblems. 

\begin{definition}\label{d:child-descendant-complete}
	A subproblem $\Gamma_u \subseteq V$ is \emph{child-descendant complete} if, for each node $v$ which is a child of a root of $M[\Gamma_u]$, each descendant of $v$ is present in $\Gamma_u$.
\end{definition}
\begin{definition}\label{d:connector-complete}
	A subproblem $\Gamma_u \subseteq V$ of multitree $M = (V,E)$ is \emph{connector complete} if, for every connector $c \in V$, if one parent of $v$ is contained in $\Gamma_u$, then \emph{all} parents of $v$ are contained in $\Gamma_u$. Formally, if any node $v \in \Gamma_u$ is connected to $c$ by an edge $(v,c)$, then for every node $v \in V$ such that $(v,c) \in E$, $v$ is also in $\Gamma_u$. 
\end{definition}
\begin{definition}\label{d:admissible}
	A subproblem $\Gamma_u \subseteq V$ is \emph{admissible} if it is both connector complete and child-descendant complete.
\end{definition}

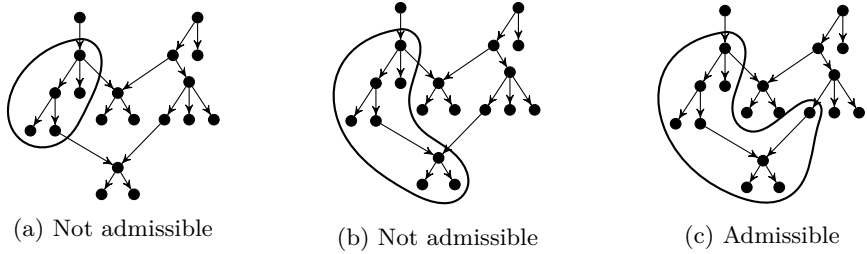
\begin{figure}
	\captionsetup[subfigure]{justification=centering}
	\centering
	\begin{subfigure}{0.3\linewidth}
		\centering
		\begin{tikzpicture}[->, >=stealth',every node/.style={minimum height=2mm, draw, circle, inner sep = 0mm,node distance=2mm, fill=black},scale=0.7,transform shape]
	
	\node [] (l) {};
	\node [right=2cm of l] (r) {};
	
	\node [below=0.5cm of l] (1) {};
	\node [below=0.5cm of 1] (2) {};
	\node [right=0.5cm of 2] (3) {};
	\node [left=0.25cm of 2] (4) {};
	\node [below= 0.5cm of 4] (5) {};
	\node [left=0.25cm of 5] (6) {};
	
	\node [below right =0.5cm of 3, xshift=-2mm] (7) {};
	\node [below left = 0.5cm of 3, xshift=2mm] (8) {};
	
	\node [below=0.5cm of r] (9) {};
	\node [left=0.25cm of 9] (10) {};
	\node [below right = 0.5cm of 10, xshift=-2mm] (11) {};
	
	\node [below = 0.5cm of 11] (12) {};
	\node [right=0.25cm of 12] (13) {};
	\node [left=0.25cm of 12] (14) {};
	
	\node [below =1.2cm of 3] (15) {};
	
	\node [below left=0.5cm of 15, xshift=2mm] (16) {};
	\node [below right=0.5cm of 15, xshift=-2mm] (17) {};

	\draw (l) edge (1)
	      (1) edge (2)
	      (1) edge (3)
	      (1) edge (4)
	      (4) edge (5)
	      (4) edge (6)
	      (3) edge (7)
	      (3) edge (8)
	      
	      (r) edge (9)
	      (r) edge (10)
	      (10) edge (3)
	      (10) edge (11)
	      
	      (11) edge (12)
	      (11) edge (13)
	      (11) edge (14)
	      
	      (5) edge (15)
	      (14) edge (15)
	      
	      (15) edge (16)
	      (15) edge (17)
		;
	\begin{scope}[on background layer]
		\draw[thick, black, >=] ([shift={(0,2mm)}]1.north) to[curve through={
			([shift={(-1mm,2mm)}]1.north west)..
			([shift={(-1mm,1)}]6.north west) ..
			([shift={(-1mm, -1mm)}]6.south west)..
			([shift={(1mm,-1mm)}]5.south east) ..
			([shift={(1mm,-1mm)}]2.south east) ..
			([shift={(1mm,2mm)}]1.north east)}] ([shift={(0,2mm)}]1.north);	

	\end{scope}	
\end{tikzpicture}
		\subcaption{Not admissible}\label{f:non-admissible-1}
	\end{subfigure}
	\hfill
	\begin{subfigure}{0.3\linewidth}
		\centering
		\begin{tikzpicture}[->, >=stealth',every node/.style={minimum height=2mm, draw, circle, inner sep = 0mm,node distance=2mm, fill=black},scale=0.7,transform shape]
	
	\node [] (l) {};
	\node [right=2cm of l] (r) {};
	
	\node [below=0.5cm of l] (1) {};
	\node [below=0.5cm of 1] (2) {};
	\node [right=0.5cm of 2] (3) {};
	\node [left=0.25cm of 2] (4) {};
	\node [below= 0.5cm of 4] (5) {};
	\node [left=0.25cm of 5] (6) {};
	
	\node [below right =0.5cm of 3, xshift=-2mm] (7) {};
	\node [below left = 0.5cm of 3, xshift=2mm] (8) {};
	
	\node [below=0.5cm of r] (9) {};
	\node [left=0.25cm of 9] (10) {};
	\node [below right = 0.5cm of 10, xshift=-2mm] (11) {};
	
	\node [below = 0.5cm of 11] (12) {};
	\node [right=0.25cm of 12] (13) {};
	\node [left=0.25cm of 12] (14) {};
	
	\node [below =1.2cm of 3] (15) {};
	
	\node [below left=0.5cm of 15, xshift=2mm] (16) {};
	\node [below right=0.5cm of 15, xshift=-2mm] (17) {};

	\draw (l) edge (1)
	      (1) edge (2)
	      (1) edge (3)
	      (1) edge (4)
	      (4) edge (5)
	      (4) edge (6)
	      (3) edge (7)
	      (3) edge (8)
	      
	      (r) edge (9)
	      (r) edge (10)
	      (10) edge (3)
	      (10) edge (11)
	      
	      (11) edge (12)
	      (11) edge (13)
	      (11) edge (14)
	      
	      (5) edge (15)
	      (14) edge (15)
	      
	      (15) edge (16)
	      (15) edge (17)
		;
	\begin{scope}[on background layer]
		\draw[thick, black, >=] ([shift={(0,2mm)}]1.north) to[curve through={
			([shift={(-1mm,2mm)}]1.north west)..
			([shift={(-1mm,1)}]6.north west) ..
			([shift={(-1mm, -1mm)}]6.south west)..
			([shift={(-1mm,-1mm)}]16.south west) ..
			([shift={(1mm,-1mm)}]17.south east) ..
			([shift={(1mm,1mm)}]15.north east) ..
			([shift={(7mm,0mm)}]5.south east) ..
			([shift={(1mm,-1mm)}]2.south east) ..
			([shift={(1mm,2mm)}]1.north east)}] ([shift={(0,2mm)}]1.north);	
	\end{scope}	
\end{tikzpicture}
		\subcaption{Not admissible}\label{f:non-admissible-2}
	\end{subfigure}
	\hfill
	\begin{subfigure}{0.3\linewidth}
		\centering
		\begin{tikzpicture}[->, >=stealth',every node/.style={minimum height=2mm, draw, circle, inner sep = 0mm,node distance=2mm, fill=black},scale=0.7,transform shape]
	
	\node [] (l) {};
	\node [right=2cm of l] (r) {};
	
	\node [below=0.5cm of l] (1) {};
	\node [below=0.5cm of 1] (2) {};
	\node [right=0.5cm of 2] (3) {};
	\node [left=0.25cm of 2] (4) {};
	\node [below= 0.5cm of 4] (5) {};
	\node [left=0.25cm of 5] (6) {};
	
	\node [below right =0.5cm of 3, xshift=-2mm] (7) {};
	\node [below left = 0.5cm of 3, xshift=2mm] (8) {};
	
	\node [below=0.5cm of r] (9) {};
	\node [left=0.25cm of 9] (10) {};
	\node [below right = 0.5cm of 10, xshift=-2mm] (11) {};
	
	\node [below = 0.5cm of 11] (12) {};
	\node [right=0.25cm of 12] (13) {};
	\node [left=0.25cm of 12] (14) {};
	
	\node [below =1.2cm of 3] (15) {};
	
	\node [below left=0.5cm of 15, xshift=2mm] (16) {};
	\node [below right=0.5cm of 15, xshift=-2mm] (17) {};

	\draw (l) edge (1)
	      (1) edge (2)
	      (1) edge (3)
	      (1) edge (4)
	      (4) edge (5)
	      (4) edge (6)
	      (3) edge (7)
	      (3) edge (8)
	      
	      (r) edge (9)
	      (r) edge (10)
	      (10) edge (3)
	      (10) edge (11)
	      
	      (11) edge (12)
	      (11) edge (13)
	      (11) edge (14)
	      
	      (5) edge (15)
	      (14) edge (15)
	      
	      (15) edge (16)
	      (15) edge (17)
		;
	\begin{scope}[on background layer]
		\draw[thick, black, >=] ([shift={(0,2mm)}]1.north) to[curve through={
			([shift={(-1mm,2mm)}]1.north west)..
			([shift={(-1mm,1)}]6.north west) ..
			([shift={(-1mm, -1mm)}]6.south west)..
			([shift={(-1mm,-1mm)}]16.south west) ..
			([shift={(1mm,-1mm)}]17.south east) ..
			([shift={(1mm,-1mm)}]14.south east) ..
			([shift={(1mm,1mm)}]14.north east) ..
			([shift={(-1mm,1mm)}]14.north west) ..
			([shift={(7mm,0mm)}]5.south east) ..
			([shift={(1mm,-1mm)}]2.south east) ..
			([shift={(1mm,2mm)}]1.north east)}] ([shift={(0,2mm)}]1.north);	
	\end{scope}	
\end{tikzpicture}
		\subcaption{Admissible}\label{f:admissible-example}
	\end{subfigure}
	\caption[Examples and non-examples of admissible subproblems.]{Each circled region denotes a subset of vertices. The subset in (a) is not child-descendant complete, (b) is not connector complete, while (c) is admissible. }\label{f:admissible-non-admissible-examples}
\end{figure}

Examples of admissible and non-admissible subproblems are shown in \autoref{f:admissible-non-admissible-examples}. Notice that, according to Definition \ref{d:admissible}, $V$ forms an admissible subproblem. This ``sub''-problem forms the root of the decomposition tree we will construct.
Our decomposition procedure decomposes each admissible subproblem into two subproblems each of which is either 1) trivial, 2) base, or 3) admissible. The decomposition is continued on admissible subproblems, while trivial and base subproblems form the leaves of the decomposition tree we will construct.

\begin{definition}\label{d:k-multitree-base-subproblem}
	A subproblem $\Gamma_u \subseteq V$ is said to be \emph{base} if $M[\Gamma_u$] forms either a $j$-multitree where $j < k$, or a trivial graph\footnote{Recall that a \emph{trivial graph} is a graph with no edges.} on $k$ nodes.
\end{definition}
Base subproblems which form $j$-multitrees for $j > 1$ are decomposed inductively by a  decomposition procedure for $j$-multitrees. Base subproblems which are $1$-multitrees are not decomposed any further. In the optimization phase, base subproblems which are $1$-multitree subproblems will be solved via the algorithm for LSP in trees presented in \cite{Mills2015}. 

Each subproblem $\Gamma_u$ is associated with a set of \emph{local roots}, which are roots of the subgraph induced by $M[\Gamma_u]$. Let $R(\Gamma_u)$ be the set of local roots of $\Gamma_u$.  Our decomposition procedure works by applying one of five cases based on the structure of the local roots and their adjacent nodes. Given a non-base, non-trivial admissible subproblem, $\Gamma_u$, the decomposition procedure uses the following recursive cases to construct a decomposition tree $\tau$.

\begin{itemize}
	\item (UP): If some local root $r \in R(\Gamma_u)$ has a single child which is not a connector, we can remove $r$ from $R(\Gamma_u)$ to form an admissible subproblem,\footnote{\label{ft:child-desc}Where admissibility follows by child-descendant completeness of $\Gamma_u$.} while $\{r\}$ forms a trivial subproblem.
	\item (OUT): If some local root $r \in R(\Gamma_u)$ has a child $c$ which has no connectors as descendants, removing $c$ and all of its descendants from $\Gamma_u$ forms an admissible subproblem.\footnoteref{ft:child-desc} Moreover, the set containing node $c$ along with its descendants forms a base subproblem.\footnoteref{ft:child-desc}
	\item (INCLUDE): If local roots in set $Q \subseteq R(\Gamma_u)$ each share a child $c$, which is the only child of each root in $Q$ and, moreover, every parent of $c$ is contained in $Q$, then we can remove the set of local roots $Q$ to form an admissible subproblem\footnoteref{ft:child-desc} $\Gamma_u \setminus Q$, while $Q$ forms a trivial subproblem.
	\item (MERGE): If every local root has one or more children and at least one local root has at least two children, then we shall show how to partition the children of each local root node along with their descendants to form two admissible subproblems $\Gamma'$ and $\Gamma''$.
\end{itemize}

To each admissible subproblem we attempt to apply each of the above cases in the order given. Only when one case does not apply are the following cases checked. 
The UP, OUT, and INCLUDE cases are each used to peel off the ``easy'' portions of the subproblem. The MERGE case is the workhorse of the decomposition, and requires additional discussion.

To partition the children of local roots in the MERGE case, we find maximal connected components in a certain hypergraph. Algorithms for finding maximal connected components in a (directed\footnote{An algorithm for undirected hypergraphs with the same running time exists. In any case, undirected hypergraphs can be handled via \cite{Allamigeon2014} by adding an extra hyperedge going in the reverse direction.}) hypergraph in $O(\alpha(N)N)$ time are known \cite{Allamigeon2014}, where $N$ is the size of the description of the hypergraph, and $\alpha(N)$ is the inverse Ackermann function. 
We will therefore constrain ourselves to discussing the hypergraph and its connection to the decomposition procedure.

In order to preserve admissibility in the MERGE case, we require that each connector from $\Gamma_u$ lie in $\Gamma'$ or $\Gamma''$ and not both. To ensure this, we form a hypergraph $H$ which has as vertices the connectors present in $\Gamma_u$, denoted by $\kappa(\Gamma_u) \subseteq \Gamma_u$. The hyperedges of $H$ are formed by the child shadows of all local roots of $\Gamma_u$. Formally, $H$ is defined via
\begin{equation}\label{eq:H-def}
H := \Big(\kappa(\Gamma_u), \bigcup_{r \in R(\Gamma_u)} \mathcal{C}(r) \Big).
\end{equation}
Thus, each hyperedge of $H$ is associated with a child of some local root of $\Gamma_u$. This association between hyperedges of $H$ and children of nodes in $R(\Gamma_u)$ is employed to further associate a subset of children of $R(\Gamma_u)$ to each strongly connected component of $H$. We form the subproblems $\Gamma'$ and $\Gamma''$ by partitioning children of $R(\Gamma_u)$ to ensure that children which fall into the same connected component of $H$ lie in the same subproblem, either $\Gamma'$ or $\Gamma''$. For example, in \autoref{f:hypergraph-example}, the children $a,b,c$ and $d$ are each associated with one maximal connected component of $H$, while the child $e$ is associated with another.

\begin{figure}
	\centering
	\hfill
	\begin{subfigure}{0.48\textwidth}
		\centering
		\begin{tikzpicture}[->, >=stealth',every node/.style={minimum height=2mm, draw, circle, inner sep = 0mm,node distance=5mm, fill=black}]

\node (r1) {};
\node [right=8mm of r1] (r2) {};
\node [right=16mm of r2] (r3) {};
\node [right=of r3] (r4) {};

\node [below = of r1, label=180:{$a~$}](b) {};
\node [right = of b, label=120:{$b~$}](a) {};
\node [right = of a, label=60:{$~c$}](d) {};
\node [right =8mm of d, label=120:{$d~$}](c) {};
\node [right = of c, label=0:{$~e$}] (e) {};

\draw (r1) edge (b);
\draw (r2) edge (a);
\draw (r2) edge (d);
\draw (r3) edge (c);
\draw (r3) edge (e);
\draw (r4) edge (e);

\node [below= 12mm of b, label={[label distance = 2mm]270:{$1$}}, xshift=4mm] (1) {};
\node [right= of 1, label={[label distance = 2mm]270:{$2$}}] (2) {};
\node [right=of 2, label={[label distance = 2mm]270:{$3$}}]  (3) {};
\node [right=of 3, label={[label distance = 2mm]270:{$4$}}]  (4) {};

\node [below=5mm of b] (b1) {};
\node [right =12mm of b1, yshift=2mm] (d1) {};
\node [right =of d1, yshift=-2mm] (c1) {};
\node [right = of c1] (c2) {};

\draw (a) edge (1);
\draw (b1) edge (1);
\draw (b) edge (b1);
\draw (b1) edge (2);
\draw (d) edge (d1);
\draw (d1) edge (3);
\draw (d1) edge (4);
\draw (c) edge (c1);
\draw (c) edge (c2);
\draw (c1) edge (1);
\draw (c1) edge (2);
\draw (c2) edge (3);
\draw (c2) edge (4);

\end{tikzpicture}
	\end{subfigure}
	\hfill
	\begin{subfigure}{0.48\textwidth}
		\centering
		\begin{tikzpicture}[every node/.style={minimum height=7mm},scale=0.5] 
\node[] (1) at (1,0) {$1$};
\node[] (2) at (3,0) {$2$};
\node[] (3) at (5,0) {$3$};
\node[] (4) at (7,0) {$4$};
\node[] (e) at (10,0) {$e$};
\node[fit=(1), shape=ellipse,draw, inner sep = -0.75mm, label={[name=lblA, label distance = -1mm]20:{$b$}}] (elps1) {};
\node[fit=(1)(2)(lblA), shape=ellipse,draw, inner sep = -0.75mm, label={[name=lblB, label distance =-1mm]20:{$a$}}, yshift=-2mm, yscale=0.95] (elps2) {};
\node[fit=(3)(4), shape=ellipse,draw, inner sep = -0.75mm, label={[name=lblD,label distance = -3mm]25:{$c$}}] (elps2) {};
\node[fit=(1)(2)(3)(4)(lblA)(lblB)(lblD), shape=ellipse,draw, inner sep = -0.75mm, label={[name=lblD, label distance =-2mm]5:{$d$}}, yshift=-2mm, xshift=-1mm,yscale=1.1] (elps3) {};
\node[fit=(e), shape=ellipse, draw, inner sep=-0.75mm, label=45:{$e$}] (elps4) {};
\end{tikzpicture}
	\end{subfigure}
	\vspace{-0.5cm}
	\caption{The hypergraph on the right depicts $H$ for the 4-multitree on the left, as defined in (\ref{eq:H-def}). Each child of a root is associated with a hyperedge which contains the connectors reachable from it. For example, $c$ is associated with the hyperedge $\{3,4\}$.}\label{f:hypergraph-example}
\end{figure}
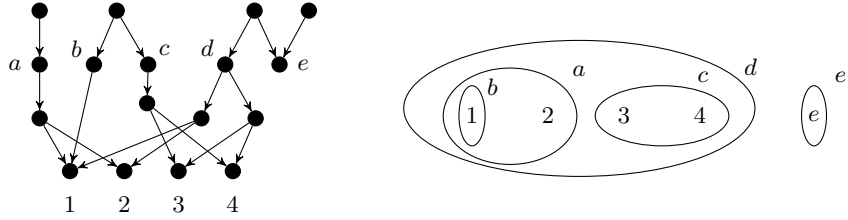


To ensure that this decomposition may be repeated as needed on the subproblems $\Gamma'$ and $\Gamma''$ we must establish a few properties of $H$.


\begin{lemma}\label{l:hypergraph-props}
	A hypergraph $H$ as defined via \eqref{eq:H-def} may be decomposed into maximal connected components $H_1 = (V_1, \mathcal{E}_1), ..., H_t = (V_t, \mathcal{E}_t)$ for which the following properties hold.
	\begin{enumerate}[i)]
		\item for all $i$, $V_i \in \mathcal{E}_i$, (i.e. each maximal connected component is covered by a single edge.)
		\item for all $i \neq j$, $V_i \cap V_j = \emptyset$, (i.e. no connector lies in two maximal connected components.)
		\item for all $r, r' \in R(\Gamma_u)$ and $i \in {1,...,t}$: $\mathcal{C}(r) \cap \mathcal{E}_i$ and $\mathcal{C}(r') \cap \mathcal{E}_i$ form a laminar pair. 
	\end{enumerate}
\end{lemma}

Statement $(iii)$ ensures that this lemma continues to hold in the subproblems $\Gamma'$ and $\Gamma''$. 
A proof of Lemma \ref{l:hypergraph-props} can be found in Appendix \ref{a:omitted-proofs}.

It remains to show that \emph{any} $k$-multitree may be decomposed according to this procedure. The proof we present here focuses on the more involved MERGE case and only sketches the argument for the INCLUDE case. A full proof appears in Appendix \ref{a:omitted-proofs}.

\begin{theorem}\label{t:k-multitree-structure-theorem}
	Any untangled $k$-multitree $M = (V,E)$ can be decomposed into a decomposition tree $\tau$ in which: 
	\begin{enumerate}[1)]
		\item all leaves of $\tau$ are associated either with base or trivial subproblems and,
		\item at each internal node $u \in V$, one of the UP, OUT, INCLUDE, or MERGE cases can be applied to the subproblem $\Gamma_u$ to obtain the subproblems associated with the children of $u$.
	\end{enumerate}
\end{theorem}

\begin{proof}
	Given an untangled $k$-multitree $M=(V,E)$, we first note that $V$ is an admissible subproblem of $G$. We proceed to show that if $\Gamma_u$ is a non-base admissible subproblem of $M$, that $\Gamma_u$ can be decomposed into at most two admissible subproblems of $G$. Since $G$ is finite, this process cannot proceed indefinitely, and thus must terminate, yielding $\tau$.
	
	If any local root $r \in R(\Gamma_u)$ has a single child which is not a connector, the UP case can be applied to yield subproblem $\Gamma_u \setminus \{r\}$. This is easily seen to be an admissible subproblem, since the child of $r$ is not a connector and $\Gamma_u$ is child descendant complete. 
	
	If some root has a child $c$ with no connectors as descendants, the OUT case can be applied as follows. The set $D$ containing $c$ and all $c$'s descendants forms a base subproblem. Thus, $\Gamma_u \setminus D$ is easily seen to be admissible.
	
	If neither the UP nor OUT case can be applied, it is clear that \begin{inparaenum}[1)] \item if any local root of $\Gamma_u$ has only a single child, it must be a connector, and \item every local root has at least one connector as a descendant. \end{inparaenum}
	Then let $c_{max}$ be the child with the maximum number of connectors as descendants. We split into two cases.
	
	{\setlength{\parindent}{0em}
		\textbf{Case 1)} Every connector in $\Gamma_u$ is a descendant of $c_{max}$.
	}	
	
	We can argue that each parent of $c_{max}$ is a local root of $\Gamma_u$ since otherwise, we can exhibit a cycle or a diamond, contradicting that $M$ is a multitree (see Appendix \ref{a:omitted-proofs}). 	
	Moreover, $c_{max}$ must have in-degree strictly greater than 1. Otherwise, it has only one parent, which implies that the UP case could be applied (a contradiction). Since the UP case cannot be applied, if $c_{max}$ has only one parent then $c_{max}$ must be a connector, which implies that $c_{max}$ has in-degree strictly greater than 1, as required.
	
	Let $Q \subseteq R(\Gamma_u)$ be the subset of local roots which are parents of $c_{max}$. Then $Q$ is a trivial subproblem while $\Gamma_u \setminus Q$ is easily seen to be an admissible subproblem on which the INCLUDE case may be applied.
	
	{\setlength{\parindent}{0em}
		\textbf{Case 2)} Some connector in $\Gamma_u$ is not a descendant of $c_{max}$.
	}
	
	In this case we apply the MERGE case by forming the hypergraph $H$ as defined in \eqref{eq:H-def}. By Lemma \ref{l:hypergraph-props}, we can form maximal connected components $H_1,...,H_t$ where $H_i = (C_i, \mathcal{E}_i)$, with $C_i \cap C_j = \emptyset$ for all $i \neq j$. To apply the MERGE case we require at least two maximal connected components, which we argue as follows.
	
	Suppose there is a single maximal connected component, $H_1 = (C_1, \mathcal{E}_1)$. By Lemma \ref{l:hypergraph-props}(i) $C_1$ is a hyperedge, which implies that there must be some child of $R(\Gamma_u)$ which covers all connectors of $\Gamma_u$. But this child must be $c_{max}$, which contradicts that some connector is \emph{not} a descendant of $c_{max}$.
	
	We can then form two admissible subproblems $\Gamma'$ and $\Gamma''$ as follows. For each local root $r \in R(\Gamma_u)$, let $X_r$ be the set of children of $r$, and let 
	\[		X_r'  := \{u \in X_r : Sh(u) \in \mathcal{E}_1 \}, ~~~~~
	X_r'' := \{u \in X_r : Sh(u) \in \mathcal{E}_2 \cup ... \cup \mathcal{E}_t \}.
	\]
	As before, since each child has at least one connector, each child is in one of $X_r'$ or $X_r''$ for some $r \in R(\Gamma_u)$.
	We form $\Gamma'$ and $\Gamma''$ as follows
	\begin{align*}
	\Gamma' &:= \{u \in \Gamma_u : u \text{ is a descendant of a node in } \bigcup_{r \in R(\Gamma_u)} X_r' \} && \hspace{-2em}\cup R(\Gamma_u);\\
	\Gamma'' &:= \{u \in \Gamma_u :  u \text{ is a descendant of a node in } \bigcup_{r \in R(\Gamma_u)} X_r'' \} &&\hspace{-2em}\cup R(\Gamma_u). 
	\end{align*}
	
	We must show that each of $\Gamma'$ and $\Gamma''$ is an admissible subproblem. Both $\Gamma'$ and $\Gamma''$ are clearly child-descendant complete, having been formed by taking all descendants of a set of children of each root.
	
	To see that $\Gamma'$ is connector complete, we will examine an arbitrary connector $c \in \Gamma'$.
	
	Since $c \in \Gamma'$, $c \in C_1$, and by Lemma \ref{l:hypergraph-props}(i), $C_1 \in \mathcal{E}_1$, which implies that there must be some node $v \in \Gamma_u$ which is a child of a local root of $\Gamma_u$ such that $Sh(v) = C_1$. Let $r \in R(\Gamma_u)$ be the local root which is a parent of $v$. Since $c$ is a connector, it must have at least two local roots as ancestors. Then let $r' \in R(\Gamma_u)$ be an arbitrary local root which is an ancestor of $c$ such that $r \neq r'$. Let $w$ be the child on the path from $r'$ to $c$. Since $M$ is untangled, and $(C_1, \mathcal{E}_1)$ is a \emph{maximal} connected component, we must have that $Sh(w) \subseteq Sh(v)$. Thus both $v$ and $w$ are in the set $\bigcup_{r \in R(\Gamma_u)} X_r'$, which implies that all of $v$ and $w$'s descendants are in $\Gamma'$, including $c$ and the two of $c$'s parents which are descendants of $v$ and $w$. Moreover, since $r'$ was chosen arbitrarily, this argument can be repeated for all $r' \in R(\Gamma_u)$ such that $r \neq r'$ to show that every parent of $c$ is contained in $\Gamma'$.
	
	A similar argument shows that $\Gamma''$ is connector complete, ending Case 2.

	Finally, the decomposition terminates since each subproblem created by this process is \emph{strictly smaller} than the subproblem from which it was formed. \qed
\end{proof}

\section{Optimizing LSP Over a Decomposition Tree}

Once the decomposition tree $\tau$ is formed via the procedure from the prior section, we can apply a recurrence bottom-up to solve $k$-LSP.

Let $\Gamma_u$ be a subproblem in decomposition tree $\tau$ which has local roots denoted by $q_1,...,q_k$. To each placement $P$ on the leaves of $M[\Gamma_u]$ we associate an \emph{ancestry signature}: a $k$-tuple in $\mathbb{N}^k$ whose $i^{th}$ entry contains the number of replicas of $P$ which have $q_i$ as an ancestor. We denote the ancestry signature of $P$ by $\vec{\alpha}(P) = \langle \alpha_1, ..., \alpha_k\rangle$.

We use the ancestry signature to index our DP recurrence, along with the number of replicas placed on a given node. We use the $F(\Gamma_u, r, \vec{\alpha})$ to denote the lexico-minimum failure aggregate obtained by any placement on the leaves of $M[\Gamma_u]$ which has size $r$ and ancestry signature equal to $\vec{\alpha}$. Since they store failure aggregates, values of $F$ are non-negative integer vectors of size $\rho + 1$. We set $F(\Gamma_u, r, \vec{\alpha}) = \infty$ when $\Gamma_u$ is a trivial subproblem, or when $M[\Gamma_u]$ does not admit any placement of size $r$ with ancestry signature $\vec{\alpha}$. We consider $\infty$ to be lexicographically larger than any vector.

Our goal is to describe $F(\Gamma_u, r, \vec{\alpha})$ in terms of values of $F$ taken the children of $u$ in subproblem tree $\tau$. Let $u$ have children $v$ and $w$. The DP recurrence we present has four cases depending on the case which was applied to $u$ to obtain $v$ and $w$. Each case of the recurrence is a sum of terms involving $\Gamma_v$ and $\Gamma_w$ along with a correction factor. This correction factor increments or decrements the number of nodes with a given failure number. Incrementing or decrementing the number of nodes with failure number $i$, is achieved by adding or subtracting $\vec{e}(i) = \langle 0,...,0,1,0,...,0\rangle$ where the 1 appears in the $(\rho - i)^{th}$ index. As we shall see, the only nodes whose failure numbers must be corrected are the local roots of subproblem $\Gamma_u$.

In the UP case, the value of $F(\Gamma_u, r, \vec{\alpha})$ must be updated to include the failure number of the new local root $q_i$. This is achieved by adding $\vec{e}(\alpha_i)$, yielding:
\[
F(\Gamma_u, r, \vec{\alpha}) = F(\Gamma_v, r, \vec{\alpha}) + \vec{e}(\alpha_i)
~~~~~~~~~ \text{(UP at root $q_i$)}.
\]

\begin{figure}[tb]
	\centering
	\begin{subfigure}{0.48\linewidth}
		\centering
		\begin{tikzpicture}[every node/.style={minimum height=2mm, draw, circle, fill=white}]

\tikzstyle{skirtnode} = [regular polygon, regular polygon sides=3,draw, minimum height=30mm, fill=none]

	\node [label=180:{$q_1$}](a) {};	
	\node [label=35:{$q_i$},right=2cm of a] (b) {};	
	\node [label=0:{$q_k$},right=of b](c) {};	
	
	\node [left=8mm of b, yshift=-3mm, minimum height=1.5mm] (out) {};
	
	\node [draw=none,fill=none] at ($(a)!0.5!(b)$) {$\dots$};
	\node [draw=none,fill=none] at ($(b)!0.5!(c)$) {$\dots$};
	
	\draw [->, >=stealth'] (b) -- (out);

	\begin{scope}[on background layer]
		\node[skirtnode, below=-2mm of a] (askirt) {};
		\node[skirtnode, below=-2mm of out, minimum height=10mm] (outskirt) {};
		\node[skirtnode, below=-2mm of b] (bskirt) {};
		\node[skirtnode, below=-2mm of c] (cskirt) {};
	\end{scope}
	
	\draw[dotted] 
		([shift={(-2mm,2mm)}]a.north west)  to[curve through={
		([shift={(1mm,1mm)}]a.north east)..
		([shift={(-1mm,-1mm)}]outskirt.south west)..
		([shift={(1mm,-1mm)}]outskirt.south east)..
		([shift={(-1mm,1mm)}]b.north west)..
		([shift={(0,2mm)}]b.north)..
		([shift={(2mm,2mm)}]c.north east)..
		([shift={(3mm,-1mm)}]cskirt.south east) ..
		([shift={(0mm,-1mm)}]bskirt.south)..
		([shift={(-3mm, -1mm)}]askirt.south west)  }] 
		([shift={(-2mm,2mm)}]a.north west);
		
	\draw [decorate, decoration={brace, amplitude=5pt}] ($(bskirt.south east)+(4mm,-1mm)$) -- ($(bskirt.south west)+(-4mm,-1mm)$) node [midway, fill=none, draw=none, yshift=-4mm] {\small $\alpha_i'$ replicas};
	
	\node [below=1.25mm of out, fill=none, draw=none] {$x$};
	
	\node [below=1.5mm of c, xshift=9mm, draw=none, fill=none] {$\Gamma_v$};
\end{tikzpicture}
	\end{subfigure}
	\hfill
	\begin{subfigure}{0.48\linewidth}
		\centering
		\begin{tikzpicture}[every node/.style={minimum height=2mm, draw, circle, fill=white}]

\tikzstyle{skirtnode} = [regular polygon, regular polygon sides=3,draw, minimum height=20mm, fill=none]

	\node [label=180:{$s_1$}](a) {};	
\node [right=of a] (b) {};	
\node [label=0:{$s_j$},right=of b](c) {};	

\node [draw=none] at ($(a)!0.5!(b)$) {$\dots$};
\node [draw=none] at ($(b)!0.5!(c)$) {$\dots$};

\node [above left=5mm of b, yshift=3mm] (d) {};
\node [above right=5mm of b, yshift=3mm] (e) {};

\node [draw=none] at ($(d)!0.5!(e)$) {$\dots$};

\draw (d) edge[->, >=stealth'] (b);
\draw (e) edge[->, >=stealth'] (b);

\begin{scope}[on background layer]
\node[skirtnode, below=-2mm of a] (askirt) {};
\node[skirtnode, below=-2mm of b] (bskirt) {};
\node[skirtnode, below=-2mm of c] (cskirt) {};
\end{scope}

\draw[dotted] 
([shift={(-2mm,2mm)}]a.north west)  to[curve through={
	([shift={(0,2mm)}]b.north)..
	([shift={(2mm,2mm)}]c.north east)..
	([shift={(3mm,-1mm)}]cskirt.south east) ..
	([shift={(0mm,-1mm)}]bskirt.south)..
	([shift={(-3mm, -1mm)}]askirt.south west)  }] 
([shift={(-2mm,2mm)}]a.north west);

\draw [decorate, decoration={brace, amplitude=5pt}] ($(bskirt.south east)+(4mm,-1mm)$) -- ($(bskirt.south west)+(-4mm,-1mm)$) node [midway, fill=none, draw=none, yshift=-4mm] {\small$\beta_\ell$ replicas};

\draw [decorate, decoration={brace, amplitude=5pt}] ($(d.north west)+(0mm,1mm)$) -- ($(e.north east)+(0mm,1mm)$) node [midway, fill=none, draw=none, yshift=4mm] {$Q$};

\node [below=1.5mm of c, xshift=9mm, draw=none, fill=none] {$\Gamma_v$};
\end{tikzpicture}
	\end{subfigure}
	\hfill
	\vspace{-1.25cm}
	\caption{Left: schematic for the OUT case; right: schematic for the INCLUDE case. Dotted lines surround $\Gamma_v$ in both cases.}\label{f:schematics-up-out}
\end{figure}
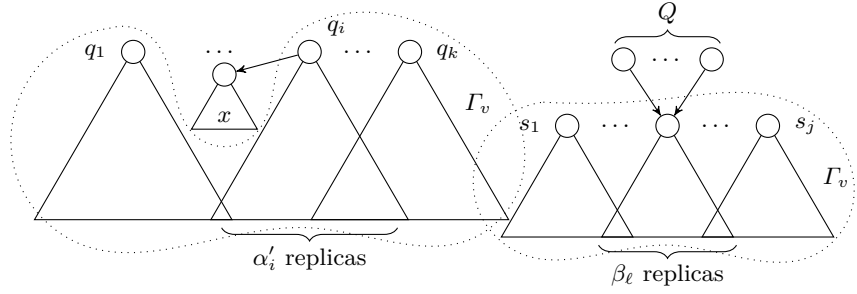

Consider next the OUT case at local root $q_i$ (see \autoref{f:schematics-up-out}). Allow $\Gamma_w$ to represent the subproblem with no connectors and recall that $M[\Gamma_w]$ forms a tree. Thus, we may use the algorithm for trees developed previously \cite{Mills2015} to find $\vec{T}(\Gamma_w, x)$ the lexico-minimum failure aggregate attainable in $M[\Gamma_w]$ using $x$. To attain the optimal value overall, we take the minimum over all possible ways to split replicas which are descendants of $q_i$ among leaves of $M[\Gamma_w]$ and $M[\Gamma_v]$.
\[
F(\Gamma_u, r, \vec{\alpha}) =
\displaystyle\min_{\substack{\alpha_i' + x = \alpha_i\\r'+x=r}}\Big[F(\Gamma_v, r', \vec{\alpha}') + \vec{T}(\Gamma_w, x) + \vec{e}(\alpha_i) - \vec{e}(\alpha_i') \Big] 
~~~~~~~~~ \text{(OUT at root $q_i$)}.
\]
where $\vec{\alpha}' := \langle \alpha_1, ...,\alpha_{i-1}, \alpha_i', \alpha_{i+1}, ..., \alpha_k\rangle$. The corrective factor of $\vec{e}(\alpha_i) - \vec{e}(\alpha_i')$ adjusts the failure number of root $q_i$ from its previous value of $\alpha_i'$ (which is included from $F(\Gamma_v, r', \vec{\alpha}')$) to its new value of $\alpha_i$.

In the MERGE case we consider subproblems $\Gamma_v$ and $\Gamma_w$ which share only the $k$ local roots among them. Thus, as in the previous case, the leaves of $M[\Gamma_v]$ and $M[\Gamma_w]$ are disjoint. Taking the lexico-minimum over all ways to split the ancestry signature $\vec{\alpha}$ into $\vec{\alpha}'$ and $\vec{\alpha}''$ yields the optimal value overall, as shown below.
\[
F(\Gamma_u, r, \vec{\alpha}) =
\displaystyle\min_{\substack{\vec{\alpha}' + \vec{\alpha}'' = \vec{\alpha}\\r' + r'' = r}}\Big[F(\Gamma_v, r', \vec{\alpha}') + F(\Gamma_w, r'', \vec{\alpha}'') + correct_k(\vec{\alpha'}, \vec{\alpha}'')\Big]
~~~~~~~~~  \text{(MERGE)}
\]
where the corrective factor $correct_k(\vec{\alpha}',\vec{\alpha}'') := \sum_{i=1}^k \vec{e}(\alpha_i) - \vec{e}(\alpha_i') - \vec{e}(\alpha_i'')$
for \mbox{$\vec{\alpha}' = \langle \alpha_1',...,\alpha_k' \rangle$} and $\vec{\alpha}'' = \langle \alpha_1'',...,\alpha_k'' \rangle$. The $i^{th}$ term in the corrective factor adjusts the failure number of root $q_i$ by replacing the contributions of $\vec{e}(\alpha_i')$ and $\vec{e}(\alpha_i'')$ (which were included from $F(\Gamma_v, r', \vec{\alpha}')$ and $F(\Gamma_w, r'', \vec{\alpha}'')$ respectively) with the corrected value of $\vec{e}(\alpha_i)$.

The INCLUDE case requires special consideration since $\Gamma_v$ has strictly fewer local roots than $\Gamma_u$. Thus placements on the leaves of $M[\Gamma_v]$ will have ancestry signatures with length $j$, whereas the parent subproblem $\Gamma_u$ requires ancestry signatures of length $k$. These signatures will need to be appropriately mapped onto one another. Moreover, not all values of $\vec{\alpha}$ are valid as ancestry signatures of $\Gamma_u$, since local roots in $Q$ must all \emph{share} the same failure number (see \autoref{f:schematics-up-out}). Thus, our recurrence will only be computed at values of  $\vec{\alpha}$ for which this is true. To address these details, we employ a mapping $\vec{h} : \mathbb{N}^j \to \mathbb{N}^k$ which maps ancestry signatures of $\Gamma_v$ to their corresponding signature in $\Gamma_u$. A formal definition of $\vec{h}$ can be found in Appendix \ref{a:omitted-proofs}.

With the mapping $\vec{h}$ in hand we can describe the optimal value of $F(\Gamma_u, r, \vec{\alpha}(\vec{\beta}))$ as follows. Let $\Gamma_v$ be the base subproblem which forms a $j$-multitree, and which has local roots $s_1,...,s_j$. Moreover, $\Gamma_v$ has a distinguished local root, $s_\ell$, whose parents all lie in the set $Q \subseteq \{q_1,...,q_k\}$. Given values of $F(\Gamma_v, r, \vec{\beta})$, we can compute the recurrence as follows
\[
F(\Gamma_u, r, \vec{h}(\vec{\beta})) =
F(\Gamma_v, r, \vec{\beta}) + |Q|\cdot\vec{e}(\beta_\ell) ~~~~~~~~~ \text{(INCLUDE where $s_\ell$ has parents in $Q$)}
\]
In the above equation, the term $|Q|\cdot\vec{e}(\beta_\ell)$ corrects for the addition of all $|Q|$ local roots in $Q$. Each such local root will have a failure number matching that of $s_\ell$. For all values of $\vec{\alpha}$ which do not match $\vec{h}(\vec{\beta})$ for some $\vec{\beta}$, we set $F(\Gamma_u, r, \vec{\alpha})) = \infty$.

\section{Time Analysis and Conclusion}

In both phases, the time required to compute the MERGE case dominates the remaining cases. To bound the time taken to run the decomposition phase, notice that the number of edges in any $k$-multitree is no more than $kn$, where $|V| = n$. Thus, the size of a description of the connector-shadow hypergraph $H$ may be no more than $O(kn)$, and therefore maximal connected components of $H$ may be found in $O(\alpha(kn)kn)$ time per application of the MERGE case. Since each application of a MERGE separates at least one connector from the rest, there may only be $O(c)$ MERGE cases, where $c$ is the number of connectors in $M$.

For the optimization phase, $O(\rho^k)$ is an upper bound on both a) the number of ways to split an ancestry signature $\vec{\alpha}$ into $\vec{\alpha}'$ and $\vec{\alpha}''$ and b) the number of values of $\vec{\alpha}$ for which $F(\Gamma_u, r, \vec{\alpha})$ must be computed. Moreover, there are $O(\rho)$ values of $r$, $O(\rho)$ ways to split values of $r$ into $r'$ and $x$, and an additional factor of $O(\rho)$ must be included for summing vector values of $F$. Overall, any MERGE phase is bounded by $O(n\rho^{2k+3})$, since each subproblem is split into two strictly smaller subproblems at each step, and this may be done only $n$ times. Notice that base subproblems considered in the INCLUDE case have strictly less than $k$ roots, so their running times are each bounded by $O(n\rho^{2j+3})$ where $j < k$. Since in practice $c$ may be either $O(n)$ or $o(\rho^{2k+3})$, we report the total running time as $O(n\rho^{2k+3} + \alpha(kn)ckn)$. A looser, somewhat snappier bound is $O(n^2\rho^{2k+3})$. Either bound suffices to establish fixed-parameter tractability of \emph{untangled} $k$-LSP.

At the end of Section 4 we briefly described how an optimal placement algorithm for untangled $k$-multitrees suffices to solve the problem in \emph{canonical placement models} and thus in DAGs. However, in the general case, the number of roots may be large, making optimization prohibitively expensive. Thus, a procedure for minimizing the number of roots in a canonical placement model would be a useful future contribution. Other directions for future work include approximation algorithms and algorithms based upon alternative parameterizations, particularly output-sensitive parameterizations based upon the failure aggregate.


\bibliographystyle{splncs03}
\bibliography{cocoa2017}


\appendix
\section{Omitted / Truncated Proofs}\label{a:omitted-proofs}

\begin{proof}[Proof of Theorem \ref{t:3-multitree-np-hard} (cont.)]
	We complete the reduction by showing that $H$ has a placement $P\subseteq H(V)$ with $|P| = k$ for which $\vec{f}(P) \leq_L \langle 3,0,...,0,\infty,\infty \rangle$ if and only if $G$ has an independent set of size $k$.
	
	$``\implies"$ Suppose $H$ has a placement $P \subseteq H(V)$ with $|P|=k$ and $\vec{f}(P) \leq \langle 3,0,...,0,\infty,\infty \rangle$. Nodes $\alpha, \beta$ and $\gamma$ each have failure number $k$, since every node in $P$ has each of $\alpha,\beta$ and $\gamma$ as an ancestor. Thus, the upper bound on $\vec{f}(P)$ implies that all other nodes in $H$ have a failure number of at most 1. Thus, no node of $H(E)$ has failure number 2, which further implies that $P$ is a subset of $k$ nodes of $H(V)$ such that no node in $H(E)$ is connected to two or more nodes of $P$. Thus, every node in $H(E)$ is connected to at most one node of $P$, which implies that no two nodes of $P$ are adjacent as vertices of $G$. Thus, $P$ corresponds to an independent set of size $|P|=k$ in $G$.
	
	$``\impliedby"$ Suppose instead that $G$ has an independent set $I$ of size $k$. Then $I$ corresponds to a subset $P \subseteq H(V)$ of size $k$ in which no two vertices of $P$ are adjacent to the same node in $H(E)$. But this implies that every node in $H(E)$ has a failure number of at most 1. Moreover, no vertex in $H(V)$ can have failure number greater than $1$, and, as we have shown, each node $\alpha, \beta$ and $\gamma$ has failure number exactly $k$. Therefore, $\vec{f}(P) \leq_L \langle 3,0,...,0,\infty,\infty \rangle$, and so $P \subseteq H(V)$ is a placement of size $k$ with the required upper bound on $\vec{f}(P)$. \qed
\end{proof} 
For convenience, Lemma \ref{l:hypergraph-props} is restated below.

\begin{numberedlemma}{\ref{l:hypergraph-props}}
	Given an admissible subproblem $\Gamma_u$, the hypergraph $H$ defined via
	\[H := \Big(\kappa(\Gamma_u), \bigcup_{r \in R(\Gamma_u)} \mathcal{C}(r) \Big)\] 
	may be decomposed into maximal connected components $H_1 = (V_1, \mathcal{E}_1), ..., H_t = (V_t, \mathcal{E}_t)$ for which the following properties hold.
	\begin{enumerate}[i)]
		\item for all $i$, $V_i \in \mathcal{E}_i$, (i.e. each maximal connected component is covered by a single edge.)
		\item for all $i \neq j$, $V_i \cap V_j = \emptyset$, (i.e. no connector lies in two maximal connected components.)
		\item for all $r, r' \in R(\Gamma_u)$ and $i \in {1,...,t}$: $\mathcal{C}(r) \cap \mathcal{E}_i$ and $\mathcal{C}(r') \cap \mathcal{E}_i$ form a laminar pair. 
	\end{enumerate}
\end{numberedlemma}
\begin{proof}[Proof of Lemma \ref{l:hypergraph-props}]
	Recall from the main body of the paper that for each $r \in R(\Gamma_u)$, any two hyperedges $X,Y \in \mathcal{C}(r)$ are disjoint. Moreover, since $M[\Gamma_u]$ is an untangled multitree, $\mathcal{C}(r)$ and $\mathcal{C}(r')$ form a laminar pair by definition. Each property may be proven as follows.
	\begin{enumerate}[(i):]
		\item Suppose that $V_i \notin \mathcal{E}_i$. Let $E$ be the largest hyperedge of $\mathcal{E}_i$, and note that $|E| < |V_i|$, since otherwise $E = V_i$. Since $H_i$ is connected, the vertices in $V_i \setminus E$ must be reachable from the vertices in $E$. These vertices can only be reached via a hyperedge of $H_i$, and since $\mathcal{E}_i$ is laminar, this vertex must entirely contain $E$, thereby contradicting that $E$ is the largest hyperedge of $H_i$.
		\item Suppose that $V_i \cap V_j \neq \emptyset$ for some $i \neq j$. Since $V_i$ is an edge of $H_i$ it must also be a hyperedge of $H$, and likewise for $V_j$. The hyperedges of $H$ are easily seen to be laminar. Thus all hyperedges of $H$ are either disjoint, or one is a subset of another. If $V_i \subseteq V_j$, then since $H_i$ is a \emph{maximal} connected component, we must have the $V_j \in \mathcal{E}_i$, and thus $V_j$ is part of the same connected component of as $V_i$, implying $i=j$, a contradiction.
		\item Since $\mathcal{E}_i$ and $\mathcal{E}_j$ form a laminar pair, and any pair of subsets of a laminar pair forms a laminar pair, in particular $\mathcal{C}(r) \cap \mathcal{E}_k$ and $\mathcal{C}(r') \cap \mathcal{E}_k$ form a laminar pair. \qed
	\end{enumerate}
\end{proof}

In the proof of Theorem \ref{t:k-multitree-structure-theorem}, we argued in Case 1 that each parent of $c_{max}$ must be a local root of $\Gamma_u$, since otherwise, we can exhibit a cycle or a diamond, both of which are forbidden structures. We now provide the justification for this claim.

\begin{proof}[Proof of Theorem \ref{t:k-multitree-structure-theorem} (cont.)]
	
	We claim that the choices for parent of $c_{max}$ are limited to nodes in $R(\Gamma_u)$. Suppose instead that some other node $v \notin R(\Gamma_u)$ is a parent of $c_{max}$. Then $v$ must have an ancestor $a \in R(\Gamma_u)$. Let $x$ be the child of $a$ on the path from $a$ to $c_{max}$. It is clear that $x$ is not a connector, as we now show. If $x$ is a connector then there exists a cycle from $c_{max} \rightsquigarrow x \rightsquigarrow c_{max}$, contradicting that $M$ is acyclic.
	
	So $x$ is not a connector. Thus, since the UP case could not be applied, root $a$ must have multiple children. Let $x' \neq x$ be another of the children of $a$. Since the OUT case(s) could not be applied, $x'$ must have a connector $y$ as a descendant. But since in Case 1 all connectors are descendants of $c_{max}$ this forms a diamond from $a \rightsquigarrow c_{max} \rightsquigarrow y$ and $a \rightsquigarrow x' \rightsquigarrow y$. Thus choices for parents of $c_{max}$ are limited to nodes in $R(\Gamma_u)$ as claimed. \qed
\end{proof}

\subsection{Mapping Ancestry Signatures in the INCLUDE Case}

Let $\Gamma_v$ be the base subproblem which forms a $j$-multitree, and which has local roots $s_1,...,s_j$. Moreover, $\Gamma_v$ has a distinguished local root, $s_\ell$, whose parents all lie in the set $Q \subseteq \{q_1,...,q_k\}$. Given values of $F(\Gamma_v, r, \vec{\beta})$, we wish to compute the optimal value of $F(\Gamma_u, r, \vec{\alpha})$ for appropriate values of $\vec{\alpha}$.

Observe that in the INCLUDE case, not all values of $\vec{\alpha}$ are valid as ancestry signatures of $\Gamma_u$, since local roots in $Q$ must all \emph{share} the same failure number. Thus, our recurrence will only be defined for values of  $\vec{\alpha}$ for which this is true. To describe this formally, we employ a mapping $\vec{h} : \mathbb{N}^j \to \mathbb{N}^k$ which maps ancestry signatures of $\Gamma_v$ to their corresponding signature in $\Gamma_u$. 

To define $\vec{h} : \mathbb{N}^j \to \mathbb{N}^k$ we employ a one-to-one mapping to capture which local roots of $\Gamma_v$ are also local roots of $\Gamma_u$. Recall that $\Gamma_v$ has local roots $s_1,...,s_j$ while $\Gamma_u$ has local roots $q_1,...,q_k$, and these roots are not necessarily distinct. Then there exists a one-to-one mapping $\pi : \mathbb{N} \to \mathbb{N}$ such that $q_i = s_{\pi(i)}$ for any local root $q_i \notin Q$.

The mapping $\pi$ allows us to formally define $\vec{h}(\vec{\beta})$ as follows. Let $\vec{\beta} = \langle \beta_1, ..., \beta_j \rangle$ be the ancestry signature of a placement on the leaves of $M[\Gamma_v]$.  For each such $\vec{\beta}$ there is one valid value of $\vec{h}(\vec{\beta})$, defined as 
\[\vec{h}(\vec{\beta}) = \langle h_1,...,h_k \rangle \text{ where } h_i = \begin{cases}
\beta_\ell & \text{ if $q_i \in Q$}, \\
\beta_{\pi(i)} & \text{ if $q_i \notin Q$}.
\end{cases}\]
A concrete example depicting how $\pi$ works together with the definitions of $\vec{\beta}$ and $\vec{h}(\vec{\beta})$ can be seen in \autoref{f:include-one-one-mapping}.

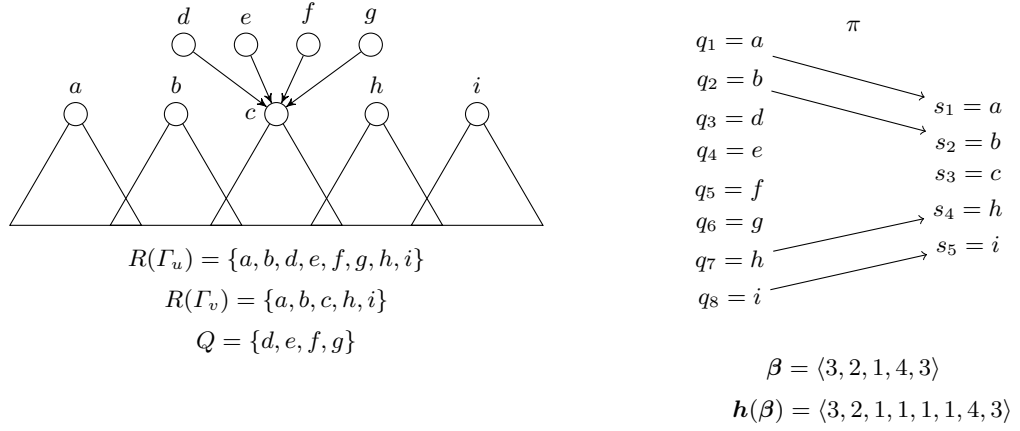
\begin{figure}[tb]
	\centering
	\begin{tikzpicture}

\node[draw, circle,fill=white, label={$a$}] (a) {};	
\node[draw, circle,fill=white, right=of a, label={$b$}] (b) {};	
\node[draw, circle,fill=white, right=of b, label=180:{$c$}] (c) {};	
\node[draw, circle,fill=white, right=of c, label={$h$}] (h) {};	
\node[draw, circle,fill=white, right=of h, label={$i$}] (i) {};	

\node[draw, circle, fill=white, above=6mm of b, label={$d$},xshift=1mm] (d) {};
\node[draw, circle, fill=white, right=5mm of d,label={$e$}] (e) {};
\node[draw, circle, fill=white, right=5mm of e,label={$f$}] (f) {};
\node[draw, circle, fill=white, right=5mm of f,label={$g$}] (g) {};

\draw[->, >=stealth'] (d) edge (c);
\draw[->, >=stealth'] (e) edge (c);
\draw[->, >=stealth'] (f) edge (c);
\draw[->, >=stealth'] (g) edge (c);

\begin{scope}[on background layer]
	\node[regular polygon, regular polygon sides=3,draw, below=-2mm of a, minimum height=20mm] {};
	\node[regular polygon, regular polygon sides=3,draw, below=-2mm of b, minimum height=20mm] {};
	\node[regular polygon, regular polygon sides=3,draw, below=-2mm of c, minimum height=20mm] {};
	\node[regular polygon, regular polygon sides=3,draw, below=-2mm of h, minimum height=20mm] {};
	\node[regular polygon, regular polygon sides=3,draw, below=-2mm of i, minimum height=20mm] {};
\end{scope}

\node[below=15mmof c] (Rgammau) {$R(\Gamma_u) = \{a,b,d,e,f,g,h,i\}$};
\node[below=0mm of Rgammau] (Rgammav) {$R(\Gamma_v) = \{a,b,c,h,i\}$};

\node[below=0mm of Rgammav] {$Q=\{d,e,f,g\}$};

\node[right=4cm of g] (q1) {$q_1 = a$};
\node[below=0cm of q1] (q2) {$q_2 = b$};
\node[below=0cm of q2] (q3) {$q_3 = d$};
\node[below=0cm of q3] (q4) {$q_4 = e$};
\node[below=0cm of q4] (q5) {$q_5 = f$};
\node[below=0cm of q5] (q6) {$q_6 = g$};
\node[below=0cm of q6] (q7) {$q_7 = h$};
\node[below=0cm of q7] (q8) {$q_8 = i$};

\node[right=8.5mm of q1, yshift=0.25cm] (pi) {$\pi$};

\node[right=2cm of q1,yshift=-0.85cm] (s1) {$s_1 = a$};
\node[below=0cm of s1] (s2) {$s_2 = b$};
\node[below=0cm of s2] (s3) {$s_3 = c$};
\node[below=0cm of s3] (s4) {$s_4 = h$};
\node[below=0cm of s4] (s5) {$s_5 = i$};

\draw (q1) edge[->] (s1);
\draw (q2) edge[->] (s2);
\draw (q7) edge[->] (s4);
\draw (q8) edge[->] (s5);

\node[below=40mm of pi, yshift=-1mm] (alpha) {$\vec{\beta} = \langle 3, 2, 1, 4, 3 \rangle$};
\node[below=0mm of alpha,xshift=2.5mm] (alpha) {$\vec{h}(\vec{\beta}) = \langle 3, 2, 1, 1, 1, 1, 4, 3 \rangle$};

\end{tikzpicture}
	\vspace{-0.5cm}
	
	\caption{The only requirement on $\pi$ is that $q_i = s_{\pi(i)}$ for any local root which is not in $R(\Gamma_u) \setminus Q$. This will occur above so long as $\pi(1) = 1,~\pi(2) = 2,~\pi(7) = 4,$ and $\pi(8) = 5$.}\label{f:include-one-one-mapping}
\end{figure}

With the mapping $\vec{h}$ in hand we can describe the optimal value of $F(\Gamma_u, r, \vec{h}(\vec{\beta}))$ by means of the recurrence
\[
F(\Gamma_u, r, \vec{h}(\vec{\beta})) =
F(\Gamma_v, r, \vec{\beta}) + |Q|\cdot\vec{e}(\beta_\ell) ~~~~~~~~~ \text{(INCLUDE where $s_\ell$ has parents in $Q$)}
\]
where the term $|Q|\cdot\vec{e}(\beta_\ell)$ corrects for the addition of all $|Q|$ local roots in $Q$. Each such local root will have a failure number matching that of $s_\ell$.

In case an ancestry signature $\vec{\alpha}'$ is not in the image of $\vec{h}$, the value of $F(\Gamma_u, r, \vec{\alpha}')$ remains $\infty$.

\end{document}